\documentclass[runningheads]{llncs}
\usepackage{amssymb,amsfonts,amsmath}
\usepackage{booktabs,tabularx}
\usepackage{url,xspace,soul,wrapfig}
\usepackage{graphicx,hyperref,paralist}
\usepackage[font=small, hypcap=true]{subfig,caption}
\usepackage{paralist,enumerate}
\usepackage[textsize=footnotesize]{todonotes}
\usepackage{enumerate,todonotes}
\usepackage{soul,lineno}
\usepackage{xcolor}
\usepackage{framed}
\usepackage{hyphenat}
\usepackage[capitalise]{cleveref}
\newcommand{\old}[1]{{}}

\usepackage{mdframed}
\usepackage{pifont}

\newcommand{\colg}[1]{{\color{green!30!black}{{\textit{#1}}}}}

\hypersetup{
hidelinks,
colorlinks=true,
citecolor=[rgb]{0.121 0.47 0.705},
linkcolor=[rgb]{0.121 0.47 0.705},
urlcolor=[rgb]{0.121 0.47 0.705}
}

\newcommand{\I}{\mathcal{I}}

\newcommand{\B}{\mathcal{B}}
\newcommand{\Sm}{\mathcal{S}}
\newcommand{\Oo}{\mathcal{O}}

\newcommand{\defn}[1]{\textit{\colg{#1}}}

\usepackage{boxedminipage}
\newcommand{\flo}[1]{\left\lfloor#1\right\rfloor}

\newcommand{\Set}{\textsc{Set Cover}\xspace}

\newcommand{\Isur}{\textsc{Minimal Interval G-SUR}\xspace}

\title{Geometric Systems of Unbiased Representatives\thanks{
Research of Sujoy Bhore is supported by the Austrian Science Fund (FWF) grant P 31119. Research of Leonardo Mart\'inez-Sandoval supported by the grant ANR-17-CE40-0018 of the French National Research Agency ANR (project CAPPS)}}

\author{Aritra Banik\inst {1}\and
Bhaswar B. Bhattacharya\inst{2}\and 
Sujoy Bhore\inst{3}\and
Leonardo Mart\'inez-Sandoval\inst{4}}
\authorrunning{A. Banik, B. Bhattacharya, S. Bhore, L. Mart\'inez}

\institute{School of Computer Sciences, National Institute of Science Education and Research, HBNI, Bhubaneswar, India
\\\email{\{aritrabanik\}@gmail.com}
\and
Department of Statistics, University of Pennsylvania, Philadelphia, USA
\\\email{\{bhaswar\}@wharton.upenn.edu}
\and
Algorithms and Complexity Group, Technische Universit\"{a}t Wien, Austria 
\\\email{\{sujoy\}@ac.tuwien.ac.at}
\and
Institut de Math\'ematiques de Jussieu-Paris Rive Gauche (UMR 7586), Sorbonne Universit\'e, France
\\\email{\{leomtz\}@im.unam.mx}
}

\begin{document}

\maketitle

\begin{abstract}
Let $P$ be a set of points in $\mathbb{R}^d$, $B$ a bicoloring of $P$ and $\Oo$ a family of geometric objects (that is, intervals, boxes, balls, etc). An object from $\Oo$ is called balanced with respect to $B$ if it contains the same number of points from each color of $B$. For a collection $\B$ of bicolorings of $P$, a geometric system of unbiased representatives (G-SUR) is a subset $\Oo'\subseteq\Oo$ such that for any bicoloring $B$ of $\B$ there is an object in $\Oo'$ that is balanced with respect to $B$.

We study the problem of finding G-SURs. We obtain general bounds on the size of G-SURs consisting of intervals, size-restricted intervals, axis-parallel boxes and Euclidean balls. We show that the G-SUR problem is NP-hard even in the simple case of points on a line and interval ranges. Furthermore, we study a related problem on  determining the size of the largest and smallest balanced intervals for points on the real line with a random distribution and coloring.

Our results are a natural extension to a geometric context of the work initiated by Balachandran et al. on arbitrary systems of unbiased representatives.




\end{abstract}

\section{Introduction}

Let $P$ be a set of size $n$. A \defn{bicoloring} $B$ of $P$ is a color assignment (red or blue) of the points in 
$P$, that is, $B:P \rightarrow \{$Red, Blue$\}$, where $B$ contains at least one red and at least one blue point.
For a bicoloring $B$, a subset of points  $P'\subseteq P$ is called \defn{balanced with respect to $B$} if $P'$ contains the same number of red and blue points, with respect to $B$. Given a set $P$ and a set of bicolorings $\B$, a \defn{system of unbiased representatives (SUR)} consists of a collection $\Sm$ of subsets of $P$ such that for every bicoloring $B \in \B$, there is at least one subset in $\Sm$ that is 
\defn{balanced} with respect to $B$. 

Balachandran et al.~\cite{DBLP:Balachandran2018} studied various problems related to finding SURs, with the motivation that SURs are useful for product testing 
over a large population. For example, suppose the effectiveness of a drug on patients is studied with respect to a large set of binary attributes related to physical characteristics, such as body weight, height, age. It is desirable to choose few families of test subjects that help to represent these attributes in a balanced manner.

Now, consider an instance where in addition we are given specific geographic locations for our test subjects and we are asked to pick them close to each other to save costs in sampling. In this situation, we cannot choose arbitrary families of test subjects: we would be required to impose some geometric constrains on them.

A natural way to model this restriction is to represent the population by a point set $P$ in Euclidean space of some dimension and to sample using ranges from some fixed family of geometric objects, that is, intervals, boxes, balls, etc. This leads to the following definitions.

For a bicoloring $B$ of $P$, we say that a geometric range is \defn{balanced with respect to $B$} if the subset of points of $P$ that it contains is balanced with respect to $B$. Given a set $P$, a set of bicolorings $\B$ and a family of allowed geometric ranges $\Oo$, 
a \defn{geometric system of unbiased representatives (G-SUR)} consists of a subfamily $\Oo'\subseteq\Oo$ such that for every bicoloring $B \in \B$, there is at least one object in $\Oo'$ that is balanced with respect to $B$.

\begin{framed}
\begin{problem} {(\textbf{G-SUR})}
Given a set $P\subset\mathbb{R}^d$ of $n$ points, a set of bicolorings $\B$ of $P$, and a collection of allowed geometric ranges  $\Oo$, find a G-SUR of minimal size.
\end{problem}
\end{framed}

For a specific attribute, it is desirable to understand how big a balanced range (for this attribute) can be.  Assuming attributes are uniformly distributed over the population, leads to the following problem:

\begin{framed}
\begin{problem}{(\textbf{Balanced Random Covering})}
Given a set $P$ of $n$ points and a random bicoloring $B$ of $P$ (chosen uniformly at random from a collection of bicolorings $\mathcal B$ of $P$), what can be said about the behavior of the size of the largest/smallest balanced interval as n goes to infinity?
\end{problem}
\end{framed}
In addition to the practical motivation, Problem~$1$ 
and Problem~$2$
are related to the vast literature
on colorings of geometric objects in which a balanced property is desired.
This includes classical results as the ham-sandwich theorem and its algoritmic version by Lo et al. \cite{lo1994algorithms}. 
Other relevant results on balanced coloring of point sets include the balanced island problem studied by Aichholzer et al. \cite{aichholzer2018computing}, balanced partitions problem for $3$-colored planar sets by Bereg et al. \cite{bereg-3},
and balanced $4$-holes in bichromatic point set \cite{bereg2015balanced}. 

\subsection{Our Results}
As an introduction to the subtleties of the geometric context, we study the G-SUR problem for $n$ points on a line and interval ranges in Section~\ref{line}. 
We show, given a set of $n$ points on a line and a collection of interval ranges, there is G-SUR of size $n-1$. Moreover, this bound is tight, that is, there are a set of bicolorings for which $n-1$ intervals are required to obtain a balanced interval (Theorem \ref{theorem-line}). 
Motivated by statistical significance, we then focus on G-SURs 
where the set of ranges are intervals of size $2k$. 
Here, we show that for any set of bicolorings $\B$, where each bicoloring in $\B$ contains more than $\flo{\frac{n}{2k}+1}(k-1)$ red and $\flo{\frac{n}{2k}+1}(k-1)$ blue points such a G-SUR exists (Theorem \ref{thm:intervalk}).
Next, for $m<n/2$, we give bounds on the size of G-SURs for when each bicoloring of $\B$ has at least $m$ red and $m$ blue points. More precisely, we show that $n-m$ intervals are always sufficient and sometimes necessary (Theorem  \ref{theorem-m-res}). All these results extend to higher dimensions to point sets in $\mathbb{R}^d$ and G-SURs consisting of axis-parallel boxes.
Section~\ref{sec:hardness} provides the hardness results. We show that the problem of finding a minimal size G-SUR is NP-hard even in the simple case of points on the real line and interval ranges (Theorem \ref{hardness}). To do this we provide a reduction from the \Set problem.

In Section \ref{high-d}, we study the problem for points in $\mathbb{R}^d$ and G-SURs consisting of Euclidean balls. Once more, we show $n-1$ balls are sometimes necessary and always sufficient to give a G-SUR (Theorem \ref{theorem-ball}).

Finally, in Section~\ref{random}, we study the Balanced Random Covering problem, where we compute the asymptotic size of the largest/smallest balanced interval for uniformly random bicolorings of points on a line in a discrete model (Theorem \ref{thm:dist}) and a continuous one (Theorem \ref{thm:distcont}).

\section{Points on a Line and Interval G-SURs}\label{line}
Let $P=\{p_1.\ldots,p_n\}$ be a set of points on the real line $\mathbb{R}$.
Throughout this section we assume that $\{p_1.\ldots,p_n\}$ is sorted from left to right on the real line.
Here our goal is to find a minimum size G-SUR  consisting of interval ranges for a given family of bicolorings $\B$ of $P$.



\subsection{Lower and Upper Bounds}
In this section we show that $n-1$ intervals are always sufficient and sometimes necessary.

\begin{theorem}\label{theorem-line}
Let $P=\{p_1,\ldots,p_n\}$ be a set of $n$ points on a line. 
Then, the following hold:

\begin{enumerate}
 \item[(a)] There exists a set of $n-1$ bicolorings $\B$, for which any G-SUR consisting of intervals has size at least $n-1$ and

 \item[(b)] There exists a set $\mathcal{I}$ of $n-1$ intervals such that for any bicoloring $B$ of $P$ there is at least one balanced interval in $\mathcal{I}$ with respect to $B$. 
\end{enumerate}
\end{theorem}
%

\begin{figure}[h]
\centering
\includegraphics[scale=0.50]{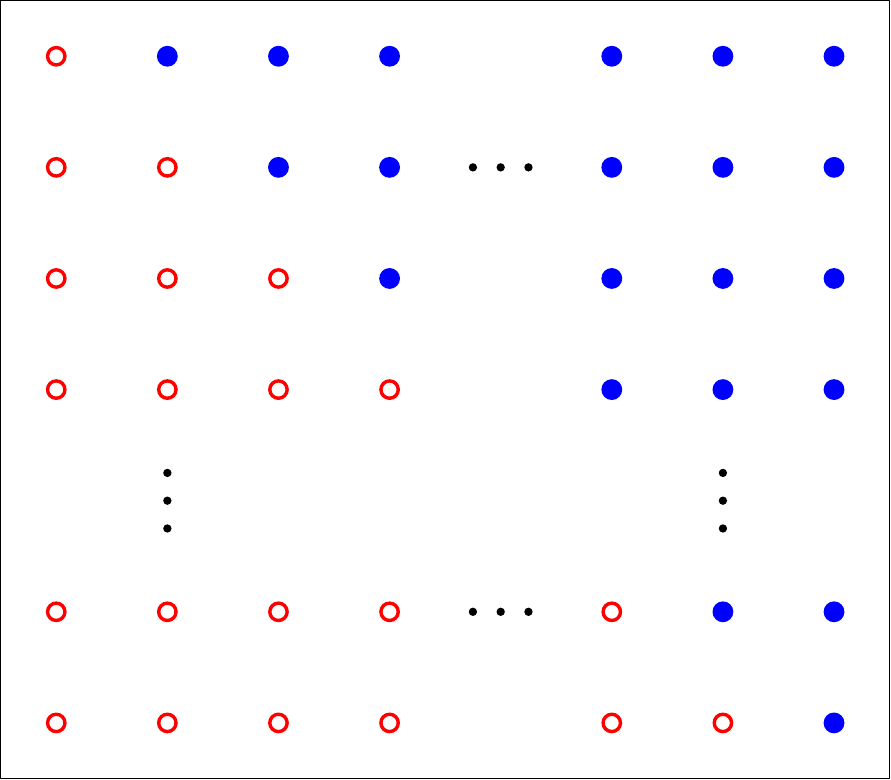}
\caption{An illustration of the case where $n-1$ intervals are necessary.} 
\label{fig:necessary-fig}
\end{figure}


\begin{proof}
We prove the first part of the theorem by constructing an example where $n-1$ intervals are necessary. Without loss of generality, we can assume  $P=\{1,2,\ldots,n\}$. We consider the set of bicolorings $\B=\{B_1,\ldots,B_{n-1}\}$ where the bicoloring $B_i$ colors the first $i$-points red and the remaining $(n-i)$ points blue (see Figure~\ref{fig:necessary-fig}).

A balanced interval with respect to $B_i$ may be shortened until its endpoints are integers. Thus, we may choose a minimal G-SUR that consists only of intervals with integral endpoints. In such a system, an interval that is balanced for $B_i$ must be symmetric around $\frac{2i+1}{2}$. Different bicolorings need intervals symmetric around different points, so a G-SUR for $\B$ requires at least $n-1$ intervals.

For the second part of the theorem consider $$\mathcal{I}=\{[p_1,p_2],[p_2,p_3],\ldots,[p_{n-1},p_n]\}.$$
Let $B$ be any bicoloring of $P$. We claim that there exists an interval in $\mathcal{I}$ which is balanced with respect to $B$. Indeed, if this is not the case then all the intervals in $\mathcal{I}$ are monochromatic, so $$B(p_1)=B(p_2)=\ldots=B(p_n),$$ which is a contradics the fact that $B$ contains at least one red and one blue point. This concludes the proof of (b).
\end{proof}

Theorem \ref{theorem-line} is already evidence of the contrast between the geometric and the abstract setting. While Balachandran et al. proved that $n-1$ arbitrary sets are sometimes necessary, their example consists of all $2^n-2$ possible bicolorings. 
The theorem above shows that only taking $(n-1)$ bicolorings are enough for the necessity, in the geometric context. 
Theorem~\ref{theorem-line} says that this is still the case even if we further restrict the allowed ranges to be intervals.

An analogous proof shows that the $n-1$ bound carries to point sets in $\mathbb{R}^d$ and G-SURs consisting of axis-parallel boxes.

\subsection{G-SURs Consisting of Intervals of Size $2k$}
In this section we fix a positive integer $k$ and consider the case where the G-SUR consists of intervals of length exactly $2k$.
Certainly, with this restriction it is not possible to have a G-SUR for every possible set of bicolorings.
For example, consider any bicoloring $B$ that contains exactly one red point. No interval of size
greater than $2$ is balanced with respect to $B$. In the following lemma we show that if each color is large enough, then there exists a G-SUR consisting of intervals of length $2k$.

\begin{theorem}\label{thm:intervalk}
Given $n\geq 2k$, a set $P=\{p_1,\ldots,p_n\}$ of points on the real line and a set of bicolorings $\B$, 
where each bicoloring in $\B$ contains more than $\flo{\frac{n}{2k}+1}(k-1)$ red and $\flo{\frac{n}{2k}+1}(k-1)$ blue points, there exist a G-SUR for $\B$
consisting of intervals of size $2k$. 
\end{theorem}  

\begin{proof}
 
Let $B$ be a bicoloring of $P$ containing more than $\flo{\frac{n}{2k}+1}(k-1)$ 
 red and $\flo{\frac{n}{2k}+1}(k-1)$ blue points.

Consider now the set of consecutive disjoint intervals of size $2k$. More formally,
$$\I=\{I_j| I_j=[p_j,p_{j+2k-1}] \text{ where } 1\leq j\leq n-2k+1\}.$$

Let $r_j$ (resp. $b_j$) be the number of red (resp. blue) points in the interval $I_j$. We say that an interval in $\I$ is \defn{red} (resp. \defn{blue}) if $r_j>b_j$ (resp. $b_j>r_j$). 

We claim that $\I$ is a G-SUR. For the sake of contradiction, we suppose that every interval from $\I$ is either red or blue. We may assume $I_1$ is red.

\begin{claim}\label{claim-red}
Every interval $I_j\in \I$ is red.
\end{claim}

\begin{proof}
We prove this by induction on $j$. By assumption, $I_1$ is red. Now, suppose that $I_j$ is red, that is, $r_j>b_j$, so $r_j\geq k+1$. As we shift from $I_j$ to $I_{j+1}$ we lose or gain at most one red point, so $r_{j+1}\geq r_{j}-1 \geq k$. It is then impossible for $I_{j+1}$ to be blue. Since each interval is either red or blue, $I_{j+1}$ must be red.
\end{proof}

Consider now the set of disjoint intervals $$\I'=\{I_j | j\equiv 1\bmod{2k}, 1\leq j\leq n-2k+1\}.$$ The intervals from $\I'$ and the interval $I_{n-2k+1}$ cover the entire set $P$. Since every interval from $\I$ is red, it has at most $k-1$ blue points. Therefore, the coloring has at most $$\flo{\frac{n}{2k}+1}(k-1)$$ blue points. This yields a contradiction to the number of blue points of $B$ given by the hypothesis, so $\I$ must have a balanced interval with respect to $B$. 
\end{proof}

The result in Theorem \ref{thm:intervalk} is tight, in the sense that if we have fewer points of either color we cannot guarantee the existence of a G-SUR of size $2k$. This can be witnessed by an example on $n=3k-1$ points on the real line and the coloring $B$ whose first $k-1$ and last $k-1$ points are blue, and the middle $k+1$ ones are red.

An analogous proof shows that Theorem \ref{thm:intervalk} extends to point sets in $\mathbb{R}^d$ and G-SURs consisting of axis-parallel boxes of size $2k$.

\subsection{G-SURs for $m$-restricted Bicolorings}
For $m\leq n/2$, a bicoloring $B$ of $P$ is \defn{$m$-restricted} if it contains 
at least $m$ points of each color. In this section we study the size of G-SURs for $m$-restricted colorings.

If $n$ is even and $m=n/2$, then there is a G-SUR of size $1$: the one consisting of the interval $[p_1,p_n]$. Otherwise, we have the following result.

\begin{theorem}\label{theorem-m-res}
Let $P=\{p_1,\ldots,p_n\}$ be a set of $n$ points on the real line and $m<n/2$ a positive integer. Then,
\begin{enumerate}
 \item[(a)] There exists a set of $n-m$ bicolorings $\B$, for which any G-SUR consisting of intervals has size at least $n-m$.

 \item[(b)] There exists a set $\mathcal{I}$ of $n-m$ intervals such that for any bicoloring $B$ of $P$ there is at least one balanced interval in $\mathcal{I}$ with respect to $B$. 
\end{enumerate}
\end{theorem}

\begin{proof}
To prove the second part of the theorem consider the set of points $$P'=\{p_1,\ldots,p_{n-m+1}\},$$ and let $\I$ be the G-SUR given in Theorem \ref{theorem-line} for $P'$, which is of size $n-m$. 
Let $B$ be any $m$-restricted bicoloring of $P$. Note that $P\setminus P'$ is of size $m-1$, so by definition it is impossible that all the red or blue points are completely contained in $P\setminus P'$. Then $P'$ contains both blue and red points, so the restriction $B'$ of $B$ to $P'$ is also a valid coloring for $P'$. We may then take a balanced interval $I\in\I$ with respect to $B'$. This interval is also balanced with respect to $B$, so $\I$ is a G-SUR for all the $m$-restricted bicolorings of $P$.

Now we prove that $n-m$ intervals are sometimes necessary. We may assume $P=\{1,2,\ldots,n\}$. We consider the set of $m$-restricted bicolorings $\B=\{B_1,\ldots,B_{n-m}\}$ defined as follows. For $i=1,2,\ldots,n-2m+1$, the bicoloring $B_i$ has the leftmost $m+i-1$ points colored red and the rest blue.  For $i=n-2m+2,\ldots,n-m$, the bicoloring $B_i$ has the points in the interval $$\{i-n+2m,\ldots,i+m-1\}$$ colored red and the remaining $m$ points colored blue. See Figure \ref{fig:mexample} for an example.

\begin{figure}[htbp]
 \centering
\includegraphics[scale=.7]{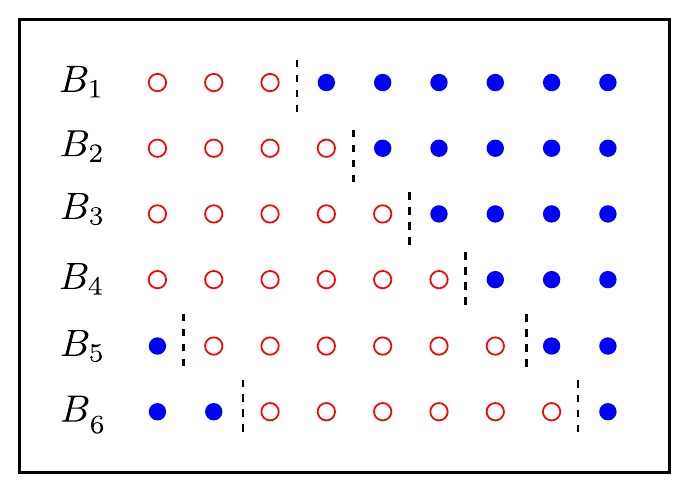}
\caption{An example of the construction in which $n-m$ intervals are needed in Theorem~\ref{theorem-m-res} for $n=9$ and $m=3$. Vertical dashed segments indicate where the symmetries must hold.}
\label{fig:mexample}
\end{figure}

As in the proof of Theorem \ref{theorem-line}, we may assume that a G-SUR consists of intervals with integral endpoints. For $i=1,2,\ldots,n-2m+1$, an interval balanced with respect to $B_i$ must be symmetric around $\frac{2m+2i-1}{2}$.

For $i\geq n-2m+2$, an interval balanced with respect to $B_i$ cannot have blue points from both the left and right sides, as otherwise it would have $n-m>n/2>m$ red points, but at most $m$ blue points. So it has to be symmetric either around the  $\frac{2i-2n+4m-1}{2}$ or around $\frac{2i+2m-1}{2}$.

Regardless of these final choices, we obtain intervals symmetric around $n-m$ different points, so they must all be different. This finishes the proof.

\end{proof}

\section{Hardness Results}\label{sec:hardness}
In this section we study the computational aspect of finding minimal size G-SURs for points on the real line using interval ranges. We receive as input a set of $n$ points $P=\{p_1,\ldots,p_n\}$ and a family of bicolorings $\B=\{B_1,\ldots,B_m\}$ of $P$. We expect as output the size of the minimal G-SUR consisting of interval ranges. We denote this problem as $\Isur$.

\begin{theorem}\label{hardness}
 The $\Isur$ problem is NP-hard. 
\end{theorem}
\begin{figure}[h]
 \centering
\includegraphics[scale=.9]{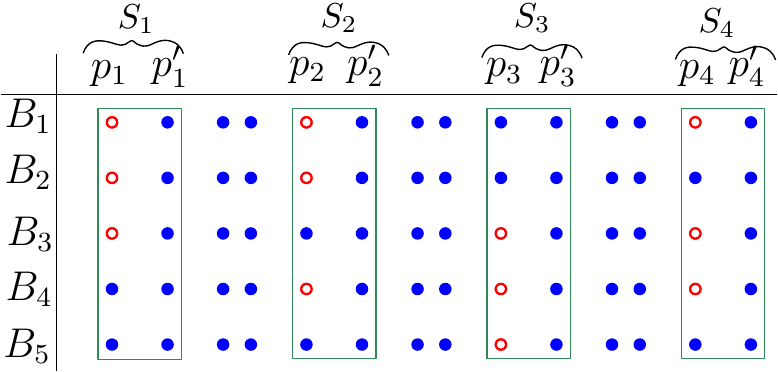}\vspace{-0.1in}
\caption{An illustration of the construction used in Theorem~\ref{hardness}. In each bicoloring $B_i$, 
the blue points between two consecutive pairs $\{p_i,p'_i\}$ 
and $\{p_{i+1},p'_{i+1}\}$ are dummy points.} \vspace{-0.1in} 
\label{fig:np}
\end{figure}

\begin{proof}
We give a reduction from the \Set problem.
In the \Set problem, we are given a set of elements $X=\{x_1,\ldots,x_n\}$ (called the universe),
and a collection $\Sm=\{S_1,\ldots,S_m\}$ of $m$ subsets, where each $S_i \subseteq X$, and $\bigcup_{1\leq i\leq m} S_i = X$.
The goal is to identify the smallest sub-collection of $\Sm$ whose union equals the universe.

From an instance $(X,\Sm)$ of the \Set problem we create an instance
$(\B,P)$ of the \Isur problem in the following manner. 
For every set $S_i$, we create a pair of consecutive points $\{p_i,p'_i\}$.
The points are in the following order 
$\{\{p_1,p'_1\}, \{p_2,p'_2\}, \ldots, \{p_m,p'_m\}\}$ on the line.
Besides for any two consecutive pair of points $\{p_i,p'_i\}$ and $\{p_{i+1},p'_{i+1}\}$, we introduce two dummy points between them. 
For every $x_j \in X$ we construct a bicoloring $B_j$ as follows: 
the points $p_i$ and $p'_i$ are colored red and blue, respectively,
if $x_j \in S_i$. Otherwise, the points $p_i$ and $p'_i$ 
are colored blue. 
Furthermore, we color all the dummy points blue in any bicoloring. 
This completes the construction.
We illustrate an instance $(\B,P)$ that is
reduced from the set system $S_1=\{x_1,x_2,x_3\}, S_2=\{x_1,x_2,x_4\},
S_3=\{x_3,x_4,x_5\}$, $S_4=\{x_1,x_3,x_4\}$ 
(see Figure~\ref{fig:np}).

In the forward direction, we show that if there is a solution of the \Set problem of size $k$
then there is a solution of \Isur problem of size $k$.
Assume that the \Set problem has a solution $\Sm^*$ where $|\Sm^*|\leq k$. We construct a set $\I^*$ of intervals cardinality $k$ as follows.
$\I^*=\{(p_i,p'_i)| S_i\in \Sm^*\}$.
The claim is $\I^*$ is a solution of \Isur problem. 
Otherwise, there exists a bicoloring $B_i\in \B$ such that there is no balanced interval in $\I^*$
with respect to $B_i$.
However we know there exists a set $S_j\in \Sm^*$ that contains $x_i$.
Thus the interval $(p_j,p'_j)$ is already taken in $\I^*$, and 
by construction it is balanced with respect to $B_i$.
Hence the claim holds.

Conversely, suppose $\I^*$ is a solution of the \Isur of cardinality at most $k$.
Observe that, due to the construction of $(\B,P)$, any interval 
that is not either of the form $(p_i,p'_i)$ or $(d,p_i)$ (where $d$ is a dummy point) 
is not balanced with respect to any $B_i$.
Hence we may assume that any interval in $\I^*$ has one of these two forms.
Next, we construct the following set $\Sm^*=\{S_j| (p_j,p'_j)\in \I^* \text{ or }
 (d,p_j)\in \I^* \}$.
Observe that $|\Sm^*|\leq k$.
We claim $\Sm^*$ is a set cover for the set system $(X,\Sm)$.
If not, there exists an element $x_i\in X$ that is not covered.
Consider the corresponding bicoloring $B_i$.
We know there is a balanced interval $I\in \I^*$ with respect to $B_i$. 
By construction, $I$ can be either $(p_j,p'_j)$ or $(d,p_j)$.
Hence we know $S_j\in \Sm^*$, and $x_i\in S_j$.
This contradicts our assumption and the claim holds.
\end{proof}

\section{Points in $\mathbb{R}^d$ and Ball G-SURs}\label{high-d}
Let $P=\{p_1.\ldots,p_n\}$ be a set of points in $\mathbb{R}^d$. Given a family of bicolorings $\B$, here the goal is to find a G-SUR ${\Oo}^*$ consisting of Euclidean balls. We show general bounds for the size of ${\Oo}^*$. 

To give a lower bound, we embed the one-dimensional example from Lemma~\ref{theorem-line} in a line $\ell$ of $\mathbb{R}^d$ and note that a ball is balanced if and only if the interval resulting from intersecting the ball with $\ell$ is balanced. Any ball creates at most one such interval on $\ell$, so a G-SUR must have size at least $n-1$.

We now show that $n-1$ balls always suffice. For this we consider the \defn{Gabriel graph} $G(P)$ whose vertex set is $P$ and there is an edge $(x,y)$ 
when the closed ball with diameter on the line segment $xy$ contains no other point of $P$.

\begin{lemma}\label{lemma:gabriel}
The Gabriel graph is connected.
\end{lemma}

This result is well-known on the plane (see e.g. \cite{de1997computational}). For completeness, here we provide a proof which works in higher dimensions.

\begin{proof}
It is enough to show that for any partition $P=Q\cup R$ of the vertices of $G(P)$ there is an edge between a vertex of $Q$ and a vertex of $R$. Let $(q,r)$ be a pair of closest points $q\in Q$ and $r\in R$, that is, that minimize $d(q,r)$.

If the ball with diameter $qr$ contains another point $r'$ from, say, $R$, then $d(q,r')<d(q,r)$, a contradiction. Similarly, this ball cannot contain another point from $Q$. So $qr$ is an edge of $G(P)$, and the proof is complete.
\end{proof}

\begin{lemma}\label{lemma:sufficient-ball}
For any set $P$ of $n$ points in $\mathbb{R}^d$ and any family $\B$ of bicolorings of $P$, 
there exists a G-SUR consisting of $n-1$ Euclidean balls.
\end{lemma}

\begin{proof}
Consider the set of $n-1$ edges $E$ of a spanning tree $T$ of the Gabriel graph $G(P)$. For every edge $e=(p,q)\in E$, let $O_e$ be the ball with diameter $pq$.
Let $\Oo=\{O_e| e\in E\}$. We claim that $\Oo$ is a G-SUR.

For any bicoloring $B$ of $P$ there is at least a red point $r$ and a blue point $b$.
Since $T$ is connected, there is a path on $T$ that connects $r$ to $b$, and thus there is an edge in this path with endpoints of opposite colors with respect to $B$. The ball corresponding to this particular edge is balanced, as it only contains $r$ and $b$.
\end{proof}

Thereby we conclude the following theorem. 

\begin{theorem}\label{theorem-ball}
Let $d\geq 1$ and $n\geq 2$ be positive integers. Then the following hold: 

\begin{enumerate}
 \item[(a)] There exists a set $P$ of $n$ points in $\mathbb{R}^d$ and a family of $n-1$ bicolorings $\B$ for $P$, for which any G-SUR consisting of Euclidean balls has size at least $n-1$.

 \item[(b)] For any set of $n$ points $P$ in $\mathbb{R}^d$ there exists a set $\Oo$ of $n-1$ Euclidean balls such that for any bicoloring $B$ of $P$ there is at least one balanced ball 
 in $\Oo$ with respect to $B$. 
\end{enumerate}

\end{theorem}

\section{Balanced Covering on Random Points on a Line}\label{random}

In this section we study the properties of balanced intervals for random bicolorings of points on a line. Consider a set $P=\{p_1, p_2, \ldots, p_{n+m}\}$ on a line, pick a subset of $P$ of size $m$ uniformly at random, color these points red. Color the remaining $n$ points blue. Define the random variables: 
\begin{itemize}
\item[--] $\mathcal{T}_{m,n}=$~the size of the smallest balanced interval.
\item[--] $\mathcal{S}_{m,n}=$~the size of the largest balanced interval,
\end{itemize}


We are interested on the asymptotic behaviour of $\mathcal{T}_{m,n}$ and $\mathcal{S}_{m,n}$ as $m$ and $n$ become large. Due to space constraints, here we focus only in the case in which $m$ is much smaller compared to $n$. In this situation we have the following result.

\begin{theorem}\label{thm:dist} For $m=o(\sqrt{n})$, 
$$\mathcal{S}_{m,n}=\mathcal{T}_{m,n}=2,$$ 
with high probability.\footnote{Here we use the usual convention that $X_n=x$ \textit{with high probability} if $\lim_{n\to\infty}\mathbb{P}(X_n=x)=1$.}
\end{theorem}


\begin{proof}

The equality $\mathcal{T}_{m,n}=2$ is direct because there are always two consecutive points of different color. 


Next, we consider $\mathcal{S}_{m,n}$. Without loss of generality, we may assume that $n>3(m+2)$. Let $E$ be the event that are at least three blue points between each pair of red points, before the first red point and after the last red point. Note that the event $\mathcal{S}_{m,n}\geq 4$ is impossible if $E$ happens. 

We can calculate the probability of $E$ happening using the following argument. Consider \defn{blocks} of colors of type $bbbrbbb$, $rbbb$ and $b$.
Each situation in which $E$ happens corresponds to placing a $bbbrbbb$ block, then $m-1$ blocks $rbbb$ to its right, and then distributing the remaining $b$'s in between these blocks. Therefore, the number of situations in which the event happens is


\[
  \binom{m+n-3(m-1)-6}{m}=\binom{n-2m-3}{m}
\]

The probability space has size $\binom{m+n}{m}$. Therefore, by Bernoulli's inequality,

\begin{align*}
\mathbb{P}(E)&=\frac{\binom{n-2m-3}{m}}{\binom{m+n}{m}}=\frac{(n-2m-3)!n!}{(n-3m-3)!(m+n)!}\\
&=\prod_{j=0}^{m-1}\left(1-\frac{3m+3}{m+n-j}\right)
             \\
             &\geq \left(1-\frac{3m+3}{n+1}\right)^{m}\\
             &\geq 1-m\left(\frac{3m+3}{n+1}\right).
\end{align*}

Since $m=o(\sqrt{n})$, the right hand side converges to $1$ as $n$ goes to infinity. This means that with high probability the event $E$ happens, and therefore, with high probability  $\mathcal{S}_{m,n}=2$. \end{proof}


In the discrete model presented above the points are equally spaced. In practical applications this is not always the case. Thus we can also study an analogous problem in the following continuous model which takes into consideration the distance between random samples.

We independently and uniformly sample $m$ points from the interval $[0,1]$ and color them red, and, similarly, sample $n$ independent and uniform points and color them blue. By symmetry, any of the red/blue discrete orderings are equally probable, and thus they distribute as in the discrete model above. Therefore, as before, $\mathcal{S}_{m,n}=\mathcal{T}_{m,n}=2$, with high probability. 
Furthermore, in this case, we can also consider the length of the balanced intervals. More precisely, 

\begin{itemize}
\item[--] $\mathcal{M}_{m,n}=$~the length of the shortest balanced interval,
\item[--] $\mathcal{L}_{m,n}=$~the length of the longest balanced interval. 
\end{itemize}

Once more, suppose that $m=o(\sqrt{n})$. Since $\mathcal{S}_{m,n}=2$ with high probability, the largest balanced interval must have as endpoints two consecutive points with high probability. Moreover, as $n$ increases, the maximum spacing between consecutive blue points converges to $0$ almost surely. These two remarks give a sketch of the proof for the following theorem.

\begin{theorem}\label{thm:distcont} For $m=o(\sqrt{n})$, 
	$\mathcal{M}_{m,n}$ and $\mathcal{L}_{m,n}$ converge to $0$ almost surely.
\end{theorem}

\clearpage
\bibliographystyle{plain}
\bibliography{main}

\begin{thebibliography}{1}

\bibitem{aichholzer2018computing}
Oswin Aichholzer, Nieves Atienza, Jos{\'e}~M D{\'\i}az-B{\'a}{\~n}ez, Ruy
  Fabila-Monroy, David Flores-Pe{\~n}aloza, Pablo P{\'e}rez-Lantero, Birgit
  Vogtenhuber, and Jorge Urrutia.
\newblock Computing balanced islands in two colored point sets in the plane.
\newblock {\em Information Processing Letters}, 135:28--32, 2018.

\bibitem{DBLP:Balachandran2018}
Niranjan Balachandran, Rogers Mathew, Tapas~Kumar Mishra, and
  Sudebkumar~Prasant Pal.
\newblock Induced-bisecting families of bicolorings for hypergraphs.
\newblock {\em Discrete Mathematics}, 341(6):1732--1739, 2018.

\bibitem{bereg2015balanced}
Sergey Bereg, Jos{\'e}~Miguel D{\'\i}az-B{\'a}{\~n}ez, R~Fabila-Monroy, Pablo
  P{\'e}rez-Lantero, A~Ram{\'\i}rez-Vigueras, Toshinori Sakai, Jorge Urrutia,
  and Inmaculada Ventura.
\newblock On balanced 4-holes in bichromatic point sets.
\newblock {\em Computational Geometry}, 48(3):169--179, 2015.

\bibitem{bereg-3}
Sergey Bereg, Ferran Hurtado, Mikio Kano, Matias Korman, Dolores Lara, Carlos
  Seara, Rodrigo~I. Silveira, Jorge Urrutia, and Kevin Verbeek.
\newblock Balanced partitions of 3-colored geometric sets in the plane.
\newblock {\em Discrete Applied Mathematics}, 181:21--32, 2015.

\bibitem{de1997computational}
Mark De~Berg, Marc Van~Kreveld, Mark Overmars, and Otfried Schwarzkopf.
\newblock Computational geometry.
\newblock In {\em Computational geometry}, pages 1--17. Springer, 1997.

\bibitem{lo1994algorithms}
Chi-Yuan Lo, Ji{\v{r}}{\'\i} Matou{\v{s}}ek, and William Steiger.
\newblock Algorithms for ham-sandwich cuts.
\newblock {\em Discrete \& Computational Geometry}, 11(4):433--452, 1994.

\end{thebibliography}

\end{document}